\documentclass[11pt]{article}
\usepackage[utf8]{inputenc}
\usepackage{mathdots}
\usepackage{amssymb}
\usepackage{amsthm}
\usepackage{graphicx}
\usepackage{amsmath}
\DeclareMathOperator{\sign}{sign}
\usepackage[capitalize]{cleveref}
\usepackage{cite}
\usepackage[margin=1in]{geometry}
\usepackage{lipsum}
\usepackage{subfig}
\usepackage{booktabs}
\usepackage[dvipsnames]{xcolor}
\usepackage{tikz}
\usepackage{epstopdf}
\usepackage{tabu}
\usepackage{array}
\usepackage{titlesec}
\usepackage{mathtools}
\DeclarePairedDelimiterX{\norm}[1]{\lVert}{\rVert}{#1}
\usepackage[title,titletoc]{appendix}

\usepackage{booktabs}
\usepackage{multirow}
\usepackage{makecell}
\usepackage[dvipsnames]{xcolor}
\makeatletter

\newcommand*{\rom}[1]{\expandafter\@slowromancap\romannumeral #1@}
\makeatother

\newtheorem{lemma}{Lemma} 
\newtheorem{theorem}{Theorem}

\usepackage{authblk}
\title{\vspace{-0.5in} Shaping the Transient Response of  Nonlinear Systems\\to Satisfy a Class of Integral Constraints}
\author[1]{Farzad Aalipour}
\author[1]{Tuhin Das}
\affil[1]{Department of Mechanical and Aerospace Engineering, University of Central Florida,}
\affil[ ]{Orlando, FL 32816, USA. E-mail: farzad.aalipour@knights.ucf.edu, tdas@ucf.edu}
\date{}    
\begin{document}

\maketitle
\vspace{-0.5in}

%%%%%%%%%%%%%%%%%%%%%%%%%%%%%%%%%%%%%%%%%%%%%%%%%%%%%%
\begin{abstract}
We consider the problem of shaping the transient step response of nonlinear systems to satisfy a class of integral constraints. Such constraints are inherent in hybrid energy systems consisting of energy sources and storage elements. While typical transient specifications aim to minimize overshoot, this problem is unique in that it requires the presence of an appreciable overshoot to satisfy the foregoing constraints. The problem was previously studied in the context of stable linear systems. A combined integral and feedforward control, that requires minimal knowledge of the plant model, was shown to make the system amenable to meeting such constraints. This paper extends that work to nonlinear systems and proves the effectiveness of the same compensation structure under added conditions. Broadly, it is shown that the integral constraint is satisfied when this compensation structure is applied to nonlinear systems with stable open-loop step response and a positive DC gain. However, stability of the resulting closed-loop system mandates bounds on the controller gain.
\vspace{-0.05in}
\end{abstract}

%%%%%%%%%%%%%%%%%%%%%%%%%%%%%%%%%%%%%%%%%%%%%%%%%%%%%
\section{Introduction} \label{sec: introduction}
\vspace{-0.05in}
In many control systems, having a desired transient response is one of the main objectives. In particular, time domain specifications, control signal magnitude, controller complexity, overshoots/ undershoots and satisfying predefined constraints, are important considerations in transient characteristics, \cite{davison2005transient}. The studies \cite{Liou1966novel,Aaron1965calculation} are among the first works on the transient response of rational transfer functions.

Different approaches to shape transient response through compensator design appear in \cite{Mohsenizadeh2011Synthesis,Kim2003Transient,moore1990}. Feedforward techniques have also been proposed to shape the transient response of linear systems for the cases of tracking an input reference. These techniques include inversion-based feedforward, \cite{GRAICHEN2005new,Kim2003Transient}. In the presence of parametric uncertainties, these methods violate the imposed constraints and can result in deterioration of performance, \cite{zhao1995feedforward}. Another issue with the inversion-based feedforward compensation is that not all of the zeros are cancellable, for instance for non-minimum phase systems, a remedy for which is discussed in \cite{clayton2009}. Inserting additional zeros in the  feedforward path for reducing the tracking errors is investigated in \cite{10.1115/1.2896362}. 

To improve the transient response of a servosystem obtained in a mode switching control, \cite{YAMAGUCHI1998Control} proposes a method of giving an additional input in the form of an impulse response. Specific transient characteristics, such as the number of extrema in step responses are studied in \cite{hauksdottir1996analytic}. In \cite{Newman2018Design}, the open-loop nature of input shaping is considered for vibration control of flexible systems, in the presence of known and finite duration disturbances. Input shaping has been extensively used for vibration control (see \cite{conker2016} and references therein). Transient response of nonlinear systems is addressed in fewer works in the literature, with some examples being  \cite{bechlioulis2010prescribed,Fan2018Asymptotic}. Constraints imposed on transient response are soft or hard. Soft constraints are similar to imposing a desirable second order performance on the controlled system, and hard constraints can be maintaining the output limited in a specific narrow range, \cite{davison2005transient}. Model predictive control is one of the approaches to handle hard constraints, \cite{mayne2001control}. This paper, like most works referenced above, deals with soft constraints.

In this paper, we study the problem of satisfying a class of integral constraints imposed on the step response of nonlinear systems. The problem is relevant in hybrid power systems where power resources and Energy Storage Systems (ESSs) are combined. For instance, in hybrid fuel cell and ultra-capacitor systems, load-following ability is directly correlated to satisfying these constraints. The load following mode for hybrid energy systems is discussed in \cite{das2012adaptive,das2011robust}. In this mode, the ESS provides or absorbs power immediately following an abrupt fluctuation in power demand, while the power source follows the load more gradually. Concurrently, a combined feedback and feedforward compensation guarantees that the ESS's State-Of-Charge (SOC) is maintained within safe limits. Fundamentally, this SOC control translates to satisfying the integral constraints mentioned above. In \cite{Salih2019AdaptiveFC,Salih2020Integral}, the compensation was studied in the context of linear systems and the concept of {\it Integral Controllability}, \cite{morari1985robust,Salih2020Integral}, was revisited. In this paper, we extend the work in \cite{Salih2020Integral} to nonlinear systems. We consider nonlinear systems that provide a stable step response at a fixed DC gain and derive conditions under which the compensated nonlinear system satisfies a desired integral constraint. 

The rest of the paper is organized as follows; In \cref{sec_back}, the constraints are explained, the viability of the aforementioned compensation structure is established for nonlinear systems and an example is given. Next, stability analysis of the compensated system is conducted in \cref{sec: scalar}, for first order plants. In \cref{sec_high}, the analysis is extended to higher order nonlinear plants and an example is provided. Subsequently concluding remarks are made and references are listed. 
\vspace{-0.05in}

%%%%%%%%%%%%%%%%%%%%%%%%%%%%%%%%%%%%%%%%%%%%%%%%%%%%%%%%%%%%
\section{Problem Definition and Extension to Nonlinear Systems}
\label{sec_back}
\vspace{-0.1in}
As discussed in the Introduction, transient responses are of importance in engineering applications such as hybrid energy systems. The problem of shaping the transient step response of a linear system to satisfy an additional integral constraint was addressed in \cite{Salih2020Integral}. In this paper, we extend the problem to nonlinear systems. An illustration of the transient response is shown in \cref{Fig: intended step response}. In this figure, $r$ is a step input and $y$ is the response of a dynamical system having a unity DC gain.
\begin{figure}[htbp]
	\begin{center} 
		\includegraphics[width=0.6\textwidth]{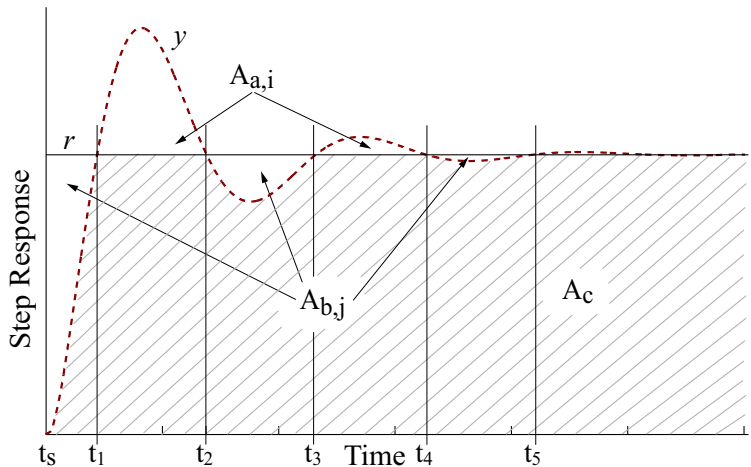}
	\end{center}
	\vspace{-0.2in}
	\caption{Transient step response with areas defining the integral constraint}
	\vspace{-0.1in}
	\label{Fig: intended step response} 
\end{figure}

There exist three different types of areas shown in \cref{Fig: intended step response}, namely $A_{a,i}$,  $A_{b,j}$, and $A_c$. The areas $A_{a,i}$ are located above the step input $r$ and below the response $y$. The areas $A_{b,j}$ are confined below $r$ and above $y$. $A_c$ refers to the common and shaded area below both $r$ and $y$. Considering these three categories, we can define 
\begin{equation*}
A_a = -\frac{1}{2} \int_{t_s}^{\infty} \Bigl\{ 1 - \sign(r-y) \Bigr\} (r - y)dt \quad \Rightarrow \quad A_a = \sum_{i=1}^{\infty} A_{a,i},
\label{eq_integ}
\end{equation*}
\begin{equation*}
A_b = \frac{1}{2} \int_{t_s}^{\infty} \Bigl\{ 1 + \sign(r-y) \Bigr\} (r - y)dt \quad \Rightarrow \quad A_b = \sum_{j=1}^{\infty} A_{b,j}
\label{eq_integ1}
\end{equation*}
and,
\begin{equation*}
A_c = \frac{1}{2} \int_{t_s}^{\infty} \Bigl\{ 1 - \sign(r-y) \Bigr\} r\,dt + \frac{1}{2} \int_{t_s}^{\infty} \Bigl\{ 1 + \sign(r-y) \Bigr\} y\,dt
\label{eq_integ2}
\end{equation*}
We can verify that
\begin{equation*}
\int_{t_s}^{\infty} y\,dt = A_{a} + A_{c}, \quad \int_{t_s}^{\infty} r\,dt = A_{b} + A_{c}.
\label{eq_integ3}
\end{equation*}
 The goal of this paper is to shape the transient step response $y$ to satisfy the integral constraint $A_b = A_a$. Such a constraint finds application in hybrid energy systems. To exemplify a practical scenario, let us assume a hybrid power system consisting of a power source and a storage unit. At an instant, if load demand increases, then the storage unit provides the extra power until the source adapts to the new power level following a transient response. In the context of \cref{Fig: intended step response}, consider the step input $r$ to represent a sudden change in power demand and $y$ to represent the response of the power source. The contribution of the storage unit can be visualized as the difference $(r - y)$. Thus, following a sudden change in demand, the storage unit provides a surge in power while the source ramps its power supply gradually. To compensate for an excess charge/discharge accrued by the storage unit during the transient, and to satisfy the load demand simultaneously, the output $y$ needs to attain $r$ at steady state while satisfying the integral constraint $A_b=A_a$ over its transient. In this case, we have from above,
\begin{equation} \label{eq_eta3}
A_a = A_b \Rightarrow \int_{t_s}^{\infty} y\,dt = \int_{t_s}^{\infty} r\,dt \Rightarrow \int_{t_s}^{\infty} e\,dt = 0 \Rightarrow e_r = 0 
\end{equation}
 where, we define $e=r-y$, and $e_r = \int_{t_s}^{\infty} e\,dt$. Additionally, practical considerations such as energy losses require that the transient response can also satisfy the following integral constraints
 \begin{equation} \label{eq_eta3a}
A_a > A_b \Rightarrow \int_{t_s}^{\infty} y\,dt > \int_{t_s}^{\infty} r\,dt \Rightarrow \int_{t_s}^{\infty} e\,dt < 0 \Rightarrow e_r < 0, \; \mbox{and similarly} \; A_a < A_b \Rightarrow e_r > 0
\end{equation}
 In this regard, it was shown in \cite{Salih2020Integral} that for linear systems, by utilizing the compensation structure of \cref{Fig: Feedback Structure} and using $\alpha$ as the inverse DC gain of the plant, the compensated system tracks any step input while satisfying the integral constraint of \cref{eq_eta3}. Further, by varying $\alpha$ the constraints of \cref{eq_eta3a} can be satisfied. In this study, we consider the plant to have nonlinear dynamics, thereby admitting more realistic models of power systems.
\begin{figure}[t]
	\begin{center} 
		\includegraphics[width=0.6\textwidth]{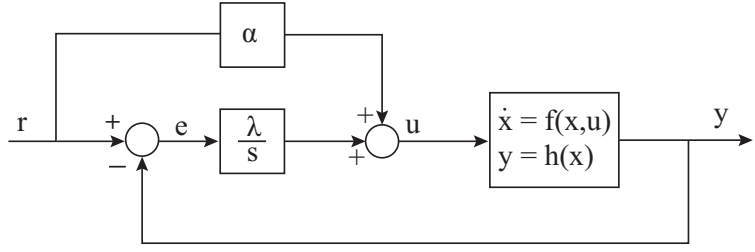}
	\end{center}
	\vspace{-0.1in}
	\caption{Compensation structure}
	\label{Fig: Feedback Structure} 
\end{figure}
We next prove that the compensation shown in \cref{Fig: Feedback Structure} also satisfies the integral constraint for nonlinear systems under certain conditions. In this context, we consider the following nonlinear system as the original open-loop nonlinear system to evaluate the compensation strategy depicted in \cref{Fig: Feedback Structure} for,   
\begin{equation} \label{eq: original nonlinear system}
\dot{x}=f(x,u), \quad y=h(x)
\end{equation}
where $x \in R^n$ is the state, $u\in R$ is the input, $y \in R$ is the output. In the closed-loop system shown in \cref{Fig: Feedback Structure}, $r$ is a step input and $\lambda \in R$ is a positive scalar. The system equations are as follows,
\begin{equation} \label{Eq: closed loop system equ}
\dot{x}=f(x,u), \quad u =\alpha r+\lambda \int_{t_s}^{t} e(\tau) d\tau, \quad e =r-y, \quad t \ge t_s
\end{equation}
We now state and prove the following lemma:
\begin{lemma}
Assume \cref{eq: original nonlinear system} has a unique and asymptotically stable equilibrium point $x_u$, for a step input $u$. In addition, when $x=x_u$, let $y=h(x_u)=u/k$, where $1/k \in R^+$ is the DC gain of \cref{eq: original nonlinear system}. Then, the equilibrium point of the closed-loop system illustrated in \cref{Fig: Feedback Structure}, with the dynamic of \cref{Eq: closed loop system equ} and $\lambda > 0$, has the following properties,
\begin{equation}
y=\frac{u}{k}=r, \quad \mbox{and} \quad  \int_{t_s}^{\infty} e(\tau) d\tau= r \left( \frac{k - \alpha}{\lambda} \right)
\label{Eq: Lemma1Equation}
\end{equation}
\label{lem1}
\end{lemma}
\begin{proof}
Considering the closed-loop system depicted in \cref{Fig: Feedback Structure} and \cref{Eq: closed loop system equ}, the state space equations of the closed-loop system are  
\begin{equation} \label{Eq: closed-loop equation}
\dot{x} = f(x,u), \quad \dot{u} = \lambda (r-y) 
\end{equation}
At equilibrium $\dot{u}=0$, implying $y=r$. Then, per the assumption in \cref{lem1}, \cref{Eq: closed loop system equ} has a unique equilibrium point at $x = x_u$ and $y=u/k$. Therefore, at steady-state, we have $y=u/k=r$ and hence from \cref{Eq: closed loop system equ}, $\int_{t_s}^{\infty} e(\tau) d\tau= (u-\alpha r)/\lambda=r(k - \alpha)/\lambda$.
\end{proof}
It is noted that in open-loop, i.e. when $\lambda = 0$, then $\int_{t_s}^{\infty} e(\tau) d\tau$ can only be determined through direct integration, as \cref{Eq: Lemma1Equation} is not applicable. The constraints of \cref{eq_eta3,eq_eta3a} are satisfied by different values of $\alpha$. The assumption of uniqueness of the equilibrium $x = x_u$ and $y = h(x_u) = u/k$ in \cref{lem1} is not restrictive. This is because there are many practical systems that are capable of tracking step inputs either naturally or through a built-in controller, such as the energy systems considered in \cite{Salih2020Integral}. The goal of this work is to design a compensation for such systems so that their transient response is shaped to satisfy the integral constraints described above. While \cref{lem1} describes the steady-state property of the compensated system in \cref{Fig: Feedback Structure}, its stability is affected by the added integral action. This stability analysis is carried out in \cref{sec: scalar,sec_high}. We next provide an example to demonstrate the result of \cref{lem1}.

\vspace{-0.05in}
%%%%%%%%%%%%%%%%%%%%%%%%%%%%%%%%%%%%%%%%%%%%%%%%%%%%%%%%%
\subsection{Example} \label{sec: example1}
\vspace{-0.05in}
We consider the following nonlinear system to demonstrate the concept discussed above
\begin{equation*}
    \dot{x} = -20(x-u)^3, \quad y = x
\end{equation*}
The above system has a unique asymptotically stable equilibrium at $x = u$ (i.e. $k = 1/\mbox{DC gain} = 1$) for a step input in $u$. When placed in the compensation structure of \cref{Fig: Feedback Structure} with $\lambda = 8$, the responses obtained for different values of $\alpha$ are shown in \cref{Fig: Constraint Example}.
\begin{figure}[htbp]
	\begin{center} 
		\includegraphics[width=\textwidth]{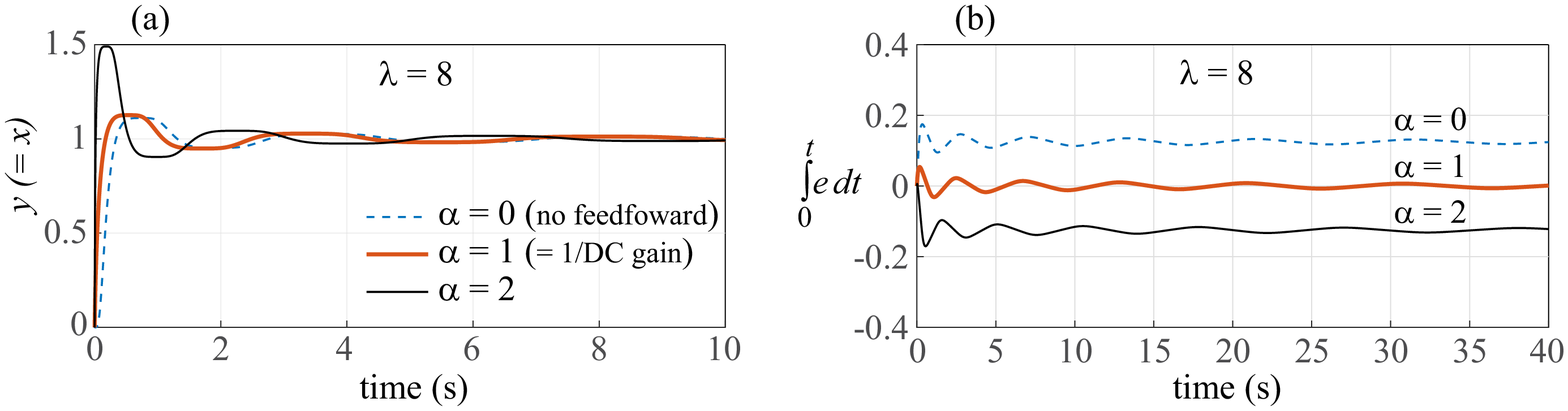}
	\end{center}
	\vspace{-0.25in}
	\caption{Illustrative example of \cref{sec: example1}}
	\vspace{-0.1in}
	\label{Fig: Constraint Example} 
\end{figure}
Figure \ref{Fig: Constraint Example}(a) shows the step response and \cref{Fig: Constraint Example}(b) shows $e_r = \int_{0}^{t} e(\tau) d\tau$. While in all cases $y$ tracks the step input in $r$, the choice $\alpha = 1 (= k)$ ensures $e_r = 0$, and $\alpha = 0$ and $2$ yield $e_r > 0$ and $< 0$ respectively.

\vspace{-0.05in}
%%%%%%%%%%%%%%%%%%%%%%%%%%%%%%%%%%%%%%%%%%%%%%%%%%%%%%%%%
\section{Stability Analysis: First Order Nonlinear System} \label{sec: scalar}
\vspace{-0.05in}
In this section, we conduct the stability analysis for the case where the original system has a first order nonlinear dynamics. The analysis will be extended to higher order nonlinear systems in \cref{sec_high}. Without loss of generality, we assume $\alpha=1$ in \cref{Fig: Feedback Structure} and state the following theorem,
\begin{theorem}
Let the nonlinear open-loop system in \cref{eq: original nonlinear system} with $y=h(x)=x$ be globally exponentially stable at its unique equilibrium point $x_u=u$, when $u$ is a step input. Then, the closed-loop system of \cref{Eq: closed-loop equation}, which represents \cref{Fig: Feedback Structure}, has a globally exponentially stable equilibrium at $[x\,\,\,u]=[r \,\,\,r]$, for sufficiently small $\lambda > 0$.  
\label{thm1}
\end{theorem}

\begin{proof}
For the open-loop first order nonlinear system, we transfer the equilibrium of \cref{eq: original nonlinear system} to the origin by defining $\bar{e}=u-x$. Considering a step input $u$, we have 
\begin{equation} \label{Eq: tranlated to closed-loop 1}
\dot{\bar{e}}=-f(u-\bar{e},u)+\dot{u}=\bar{f}(\bar{e},u)
\end{equation}
Per the assumption in \cref{thm1}, the origin $\bar{e} = 0$ is a globally exponentially stable equilibrium in \cref{Eq: tranlated to closed-loop 1}. Hence, using Converse Lyapunov Theorem in \cite{khalil2002nonlinear}, there exists a Lyapunov function $V(\bar{e})$ such that,
\begin{subequations}
\begin{gather}
c_1 |\bar{e}|^2 \leq V(\bar{e}) \leq c_2 | \bar{e}|^2 \label{Eq: upper and lower limits of Vbar}\\
\frac{\partial V}{\partial \bar{e}}\bar{f}(\bar{e},u) \leq -c_3 |\bar{e}|^2 \label{Eq: derivative upper limit}\\
\left|\frac{\partial V}{\partial \bar{e}}\right| \leq c_4 |\bar{e}| \label{Eq: derivative limit}
\end{gather}
\end{subequations}
where, $c_1, \, c_2, \, c_3, \, c_4 > 0$. We note that $V = V(\bar{e})$ (i.e. $V$ has no dependence on $u$) is not restrictive and can be defined for many non-autonomous systems, such as the ones to be discussed in \cref{Affine}. When the system of \cref{eq: original nonlinear system} is compensated as shown in \cref{Fig: Feedback Structure}, $r$ is the step input and $\dot{u} \neq 0$. Hence, the dynamic of $\bar{e}$ takes the form, $\dot{\bar{e}}=-f(u-\bar{e},u)+\dot{u}=\bar{f}(\bar{e},u)+\dot{u}$. Additionally, per \cref{Eq: closed loop system equ}, $e=r-x$, thus $\dot{e}=-\dot{x}$. Hence, from \cref{Eq: closed-loop equation}, we have the following closed-loop equation,
\begin{subequations} \label{Eq: change of variable for eqs 45}
	\begin{align}
	&\dot{\bar{e}}=-f(u-\bar{e},u)+\dot{u}=\bar{f}(\bar{e},u)+\dot{u}\\
	&\dot{u}=\lambda e
	\end{align}
\end{subequations}
To prove the stability of \cref{Eq: change of variable for eqs 45}, the following positive definite function $\bar{V}(\bar{e},e)$ is chosen
\begin{equation} 
\bar{V}(\bar{e},e)=V(\bar{e})+\frac{1}{2}(u-r)^2 = V(\bar{e})+\frac{1}{2}(\bar{e}-e)^2
\label{Eq: mod lyap func}
\end{equation}
In \cref{Eq: mod lyap func}, we note that $u-r=u-x+x-r=\bar{e}-e$. The derivative of $\bar{V}(\bar{e},e)$ along the trajectories 
of \cref{Eq: change of variable for eqs 45} is
\begin{equation*}
\begin{aligned}
\dot{\bar{V}}=\frac{\partial V}{\partial \bar{e}}\dot{\bar{e}} + (u-r)(\dot{u}-\dot{r})= 
\frac{\partial V}{\partial \bar{e}} \left[ \bar{f}(\bar{e},u)+\dot{u} \right]  +\lambda e (\bar{e}-e)  
\end{aligned}
\end{equation*}
Then,
\begin{equation*}
\begin{aligned}
&\dot{\bar{V}}=\frac{\partial V}{\partial \bar{e}} \bar{f}(\bar{e},u)+\frac{\partial V}{\partial \bar{e}}\dot{u}+\lambda e (\bar{e}-e)  
\end{aligned}
\end{equation*}
Therefore, from \cref{Eq: derivative upper limit,Eq: derivative limit},
\begin{equation} \label{Eq: proof inequlity}
\dot{\bar{V}} \leq -c_3 |\bar{e}|^2-\lambda |e|^2+\lambda(1+c_4)|e| |\bar{e}|
\end{equation}
We can express \cref{Eq: proof inequlity} as 
\begin{equation} \label{Eq: matrix format inequlity}
\dot{\bar{V}} = - \left[\begin{array}{cccc}
\bar{e}& e \end{array}\right] \left[\begin{array}{cc} c_3& \frac{\lambda(1+c_4)}{2}\\\frac{\lambda(1+c_4)}{2} &\lambda    \end{array}\right] \left[\begin{array}{cccc}
\bar{e}\\ e \end{array}\right] = -E^T \, Q \, E, \quad E = \left[ \bar{e} \quad e \right]^T
\end{equation}
By choosing $0<\lambda<4c_3/(1+c_4)^2$, $Q$ becomes positive definite. Let $p_{min}$ and $p_{max}$ be the minimum and maximum eigenvalues of $Q$, respectively. Also, it can be shown from \cref{Eq: upper and lower limits of Vbar,Eq: mod lyap func} that $\bar{V} < (c_2 + 1) \lVert E \rVert^2$. Thus,
\begin{equation} \label{Eq: negative definite matrix}
\dot{\bar{V}} = - E^T \, Q \, E \le -p_{min} \lVert E \rVert^2 \le -\frac{p_{min}}{c_2+1} \bar{V}
\end{equation}
where, from Comparison Lemma \cite{khalil2002nonlinear}, $\bar{V}(t)\leq \bar{V}(0)e^{-\frac{p_{min}}{c_2 + 1}t}$. Therefore, the origin $E = [0]$, i.e. $x = u = r$, is globally exponentially stable if $0<\lambda<4c_3/(1+c_4)^2$.
\end{proof}

%%%%%%%%%%%%%%%%%%%%%%%%%%%%%%%%%%%%%%%%%%%%%%%%%%%

Next, we study the closed-loop stability of \cref{Eq: closed-loop equation} when the equilibrium $x=x_u=u$ of the original nonlinear system, defined in \cref{eq: original nonlinear system}, is asymptotically stable for step inputs in $u$. 

\begin{theorem} 
Let the nonlinear system of \cref{eq: original nonlinear system} have a unique globally asymptotically stable equilibrium at $x_u=u$, when $u$ is a step input. Then the closed-loop system of \cref{Eq: closed-loop equation}, which is represented by \cref{Fig: Feedback Structure}, has a unique globally asymptotically stable equilibrium at $[x\,\,\,u]=[r \,\,\,r]$ for sufficiently small $\lambda > 0$, if the class $\mathcal{K}$ functions $\alpha_3$ and $\alpha_4$ in \cref{Eq: asymptotic derivative upper limit,Eq: asymptotic derivative limit}, satisfy \cref{Eq: numerator}.
\end{theorem}

\begin{proof}
With $\bar{e}=u-x$, for a step input in $u$, $\dot{u} = 0$, and thus \cref{eq: original nonlinear system} can be expressed as
\begin{equation} \label{Eq: Asymptotic system}
    \dot{\bar{e}}=-f(u-\bar{e},u)+\dot{u}=\bar{f}(\bar{e},u), \quad \mbox{where}, \quad \bar{f}(\bar{e},u) = -f(u-\bar{e},u)
\end{equation}
Based on the assumption that the original system is asymptotically stable at its unique equilibrium point $x = x_u = u$, \cref{Eq: Asymptotic system} is globally asymptotically stable at the origin $\bar{e} = 0$. Therefore, from converse Lyapunov stability theorems in \cite{khalil2002nonlinear}, there exists a Lyapunov function $V$ such that
\begin{subequations}
\begin{gather}
\alpha_1(|\bar{e}|)  \leq V(\bar{e}) \leq \alpha_2(|\bar{e}|) \label{Eq: asymtotic upper and lower limits of Vbar}\\
\frac{\partial V}{\partial \bar{e}}\bar{f}(\bar{e},u) \leq-\alpha_3(|\bar{e}|)\label{Eq: asymptotic derivative upper limit}\\
\left|\frac{\partial V}{\partial \bar{e}}\right| \leq \alpha_4 (|\bar{e}|) \label{Eq: asymptotic derivative limit}
\end{gather}
\end{subequations}
where $\alpha_1, \alpha_2, \alpha_3$ and $\alpha_4$ are class $\mathcal{K}$ functions. We define the following positive definite, radially unbounded Lyapunov function for the closed-loop system, $\bar{V}=V(\bar{e})+\frac{1}{2}(r-u)^2$. For the closed-loop system, $\dot{u} \ne 0$. Therefore, taking the derivative of $\bar{V}$ along system trajectories for a step input in $r$ we get,
\begin{equation} \label{Eq: Asymptotic Lyapunov}
\dot{\bar{V}}= \frac{\partial V}{\partial \bar{e}} \bar{f}(\bar{e},u) + \frac{\partial V}{\partial \bar{e}}\dot{u}+\lambda e(u-r) 
\end{equation}
Noting that $e = \left( \bar{e} + (r-u) \right)$ and $\dot{u} = \lambda e = \lambda \left( \bar{e} + (r-u) \right)$, \cref{Eq: Asymptotic Lyapunov} can be rewritten as 
\begin{equation} \label{Eq: Asumptotic Case}
\begin{aligned}
\dot{\bar{V}} &= \frac{\partial V}{\partial \bar{e}}\bar{f}(\bar{e},u) +\frac{\partial V}{\partial \bar{e}} \lambda (r-u)+ \frac{\partial V}{\partial \bar{e}} \lambda \bar{e} -\lambda (r-u)^2+\lambda \bar{e} (u-r) \\
&\leq -\alpha_3(|\bar{e}|)+ \alpha_4(|\bar{e}|)\lambda |r-u|+ \alpha_4(|\bar{e}|) \lambda |\bar{e}|-\lambda (r-u)^2 +\lambda |\bar{e}||r-u| 
\end{aligned}
\end{equation}
Equation (\ref{Eq: Asumptotic Case}) can be expressed as
\begin{equation*}
\dot{\bar{V}}\leq -\Bigl[\alpha_3(|\bar{e}|)- \alpha_4(|\bar{e}|)\lambda |\bar{e}|\Bigr]-\lambda (r-u)^2+\Bigl[\alpha_4(|\bar{e}|)+|\bar{e}|\Bigr]\lambda |r-u|
\end{equation*}
The expression above is quadratic in $|r-u|$. Hence, to ensure $\dot{\bar{V}} < 0$ we can impose,
\begin{equation*} \label{Eq: b2-4ac}
\lambda^2 \Bigl[\alpha_4(|\bar{e}|)+|\bar{e}|\Bigr]^2 - 4\lambda \Bigl[\alpha_3(|\bar{e}|)-\alpha_4(|\bar{e}|)\lambda |\bar{e}| \Bigr]  < 0 \quad \forall \; \bar{e} \ne 0,
\end{equation*}
which holds if $\lambda$ satisfies
\begin{equation} \label{eq:lambda limits}
0 < \lambda < \frac{\alpha_{3}(|\bar{e}|)} {\frac{1}{4}\Bigl[\alpha_{4}(|\bar{e}|) + |\bar{e}|\Bigr]^2 + \alpha_{4}(|\bar{e}|)|\bar{e}|}
\end{equation}
The right hand side of the above inequality must have a lower limit for all $\bar{e}$, so that a fixed upper limit of $\lambda$ can be defined for all $\bar{e}$. To this end, we impose that there exists a constant $\sigma > 0$ such that the following holds for all $|\bar{e}|$,
\begin{equation} \label{Eq: numerator}
    0 < \sigma < \frac{\alpha_{3}(|\bar{e}|)} {\frac{1}{4}\Bigl[\alpha_{4}(|\bar{e}|) + |\bar{e}|\Bigr]^2 + \alpha_{4}(|\bar{e}|)|\bar{e}|} \quad \Rightarrow \quad \alpha_{3}(|\bar{e}|) > \sigma \Bigl\{ \frac{1}{4}\Bigl[\alpha_{4}(|\bar{e}|) + |\bar{e}|\Bigr]^2 + \alpha_{4}(|\bar{e}|)|\bar{e}| \Bigr\}
\end{equation}
If the class $\mathcal{K}$ functions $\alpha_3$ and $\alpha_4$ satisfy \cref{Eq: numerator}, then choosing $0 < \lambda < \sigma$ ensures $\dot{\bar{V}} < 0$ and hence guarantees asymptotic stability of the equilibrium $x=u=r$.
\end{proof}
We end this section by noting that the proposed compensation of \cref{Fig: Feedback Structure} is one of many possible candidate compensators, including ones with more complex controllers as well as ones with pre-filters. However, an advantage of the compensation of \cref{Fig: Feedback Structure} is its simplicity.

%%%%%%%%%%%%%%%%%%%%%%%%%%%%%%%%%%%%%%%%%%%%%%%%%%%%%
\subsection{First Order Affine Nonlinear System} \label{Affine}

We now consider a special case, where the first order nonlinear system is affine in $u$ and has a unique exponentially stable equilibrium at $x = x_u = u$ for step inputs. The system dynamics is,
\begin{equation} \label{Eq: affine format of system}
\dot{x} = f(x,u) = h(x)+g(x)u
\end{equation}
Since $x = u$ is a unique equilibrium, we have
\begin{equation}
   h(u)+g(u)u=0 \; \Rightarrow \; h(u)=-g(u)u \; \Rightarrow \; h(x)=-g(x)x \quad \forall \; x \in R
\end{equation}
Therefore, \cref{Eq: affine format of system} can be written as
\begin{equation} \label{Eq: simplified format of fxu}
	\dot{x}=-g(x)(x-u)
\end{equation}
The uniqueness of the equilibrium at $x=u$ implies $g(x) \ne 0 \; \forall \; x \in R$. Additionally, if $f(x,u)$ is continuous in $x$, then so is $g(x)$. Since $g(x)$ is non-zero, its continuity implies that it is either positive or negative for all $x\in R$. If $g(x)$ satisfies the bounds $0 < b \le g(x) \le a$, then by defining the Lyapunov function candidate $V = 0.5\bar{e}^2$, where $\bar{e} = (u-x)$, we have
\begin{equation}
    \dot{V} = -\dot{x}\bar{e} = -g(x)\bar{e}^2 \le -b\bar{e}^2
\end{equation}
Therefore, from the Theorem 4.10 of  \cite{khalil2002nonlinear}, we conclude global exponential stability of $\bar{e} = 0$. Referring to \cref{Fig: Feedback Structure}, the closed-loop state space dynamic is
\begin{equation} \label{Eq: overall state space}
\begin{aligned}
\dot{x}= -g(x)(x-u)\\
\dot{u}= \lambda (r-x)
\end{aligned}
\end{equation}
From \cref{thm1}, we conclude that for sufficiently small $\lambda > 0$, the closed-loop system of \cref{Eq: overall state space} will be exponentially stable at $[x\,\,\,u]=[r \,\,\,r]$. For the affine system of \cref{Eq: affine format of system} and the Lyapunov function $V$ above, we can determine values of parameters $c_1$, $c_2$, $c_3$ and $c_4$, appearing in \cref{thm1}.

%%%%%%%%%%%%%%%%%%%%%%%%%%%%%%%%%%%%%%%%%%%%%%%%%%%%%%%
\section{Extension to Higher Order Systems}
\label{sec_high}
In this section, we investigate integral constraint of \cref{eq_eta3} applied to higher order nonlinear SISO systems. We consider general higher order SISO nonlinear systems given by
\begin{equation} \label{General Nonlinear System}
\dot{x}=f(x,u), \quad  y=h(x)
\end{equation}
where $x \in R^n$, $f(x,u): R^n \times R \rightarrow R^n$ and $h(x): R^n \rightarrow R$ are sufficiently smooth functions. We assume this nonlinear system has a definite relative degree $\rho$ satisfying $1\leq \rho \leq n$, \cite{Seborg1996Nonlinear}, and that for any step input $u\in R$, the system has a unique equilibrium at $x=x_u$ such that $f(x_u,u)=0$ and $\lim_{t \rightarrow \infty} y=h(x_u)= u$. Equation (\ref{General Nonlinear System}) is expressed in the {\it normal form} \cite{khalil2002nonlinear}, using the states $\xi = [\xi_1, \, \xi_2, \cdots \, \xi_\rho]^T$ and $\eta = [\eta_1, \, \eta_2, \cdots \, \eta_{n - \rho}]^T$ as follows,
\begin{equation} \label{Eq: Linearized form of Nonlinear System}
\begin{aligned}
y \hspace{1mm}&= h(x) = \xi_1\\
\dot{\xi}_1 &= L_fh= \xi_2 \\
\dot{\xi}_2 &= L^2_f h =\xi_3\\
& \,\,\,\,\, \vdots\\
\dot{\xi}_{\rho} & = L^{\rho}_f h=f_1(\xi,\eta,u)\\
\dot{\eta} \hspace{1mm} &= f_0(\xi,\eta, u) \in R^{(n-\rho)}
\end{aligned}
\end{equation}
where $L_fh(x)$ is the Lie derivative of function $h(x)$ along $f(x,u)$. We assume that the mapping between $x$ and $[\xi^T, \,\, \eta^T]^T$ represents a valid diffeomorphism. In \cref{Eq: Linearized form of Nonlinear System}, $\xi_i$ for all $i=1,2,\cdots,\rho$ represent a chain of integrators, and $\eta \in R^{(n - \rho)}$ represent the internal states. From \cref{Eq: Linearized form of Nonlinear System}, at steady state
\begin{equation}
\xi_1 = u, \,\, \xi_2=0, \,\, \cdots \,\, \xi_{\rho}=0 
\end{equation}
Further, we assume that the internal dynamic $\dot{\eta} = f_0(\xi,\eta, u)$ is stable for any step input in $u$, and $\lim_{t \rightarrow \infty} \eta = \eta_u$. The nonlinear system of \cref{Eq: Linearized form of Nonlinear System} can be expressed as
\begin{equation}
    \begin{aligned}
     & \dot{\xi} =A_c \xi + B_c f_1 (\xi,\eta,u)\\
     & \dot{\eta} =f_0(\xi, \eta, u)\\
     & y =C_c \xi
    \end{aligned}
\end{equation}
where
\begin{equation}
A_c = \left[
\begin{array}{ccccc}
0 & 1 & 0 & \cdots & 0 \\
0 & 0 & 1 & \cdots & 0 \\
\vdots & & \ddots & & \vdots \\
\vdots & & & 0 & 1 \\
0 & \cdots & \cdots & 0 & 0
\end{array}
\right], 
\quad B_c= \left[ \begin{array}{c} 0 \\ 0 \\ \vdots \\ 0 \\ 1 \end{array} \right],
\quad C_c=\left[ \begin{array}{ccccc} 1&0&\hdots&0&0 \end{array} \right]
\label{eq_normalform}
\end{equation}
Defining
\begin{equation}
  \left[ \begin{array}{c} \bar{e}\\\bar{e}_{\eta} \end{array} \right]= \left[ \begin{array}{c} u\\0\\\vdots\\0\\---\\\eta_u \end{array} \right]- \left[ \begin{array}{c} \xi\\ --- \\ \eta \end{array} \right] 
\end{equation}
we have,
\begin{equation}
 \left[ \begin{array}{c} \dot{\bar{e}}\\\dot{\bar{e}}_{\eta} \end{array} \right] = \left[ \begin{array}{c} C^T_c \\---\\ d \eta_u/{d u} \end{array} \right]\dot{u}\, - \, \left[ \begin{array}{c} \dot{\xi}\\ ---\\ \dot{\eta} \end{array} \right]
\end{equation}
Here, $\bar{e} = [\bar{e}_1, \, \bar{e}_2, \, \cdots,  \bar{e}_\rho]^T$, $\bar{e}_{1}=u-\xi_1=u-y$, $\bar{e}_{2}=-\xi_2$, $\cdots$, $\bar{e}_{\rho}=-\xi_{\rho}$. Further, the term $d \eta_u/{d u}$ represents the rate of change of the steady-state $\eta_u$ with $u$. Since we consider $u$ is a step input, therefore $\dot{u}=0$ and $\dot{\bar{e}}_{1}= \bar{e}_{2}$, $\dot{\bar{e}}_{2}=\bar{e}_{3}$, $\cdots$, $\dot{\bar{e}}_{\rho} = -\dot{\xi}_\rho = - f_1(\xi,\eta,u) = -\bar{f}_1(\bar{e},\bar{e}_\eta,u)$,
$\dot{\bar{e}}_{\eta} = - \dot{\eta} = - f_0(\xi,\eta,u) = \underline{f_0} (\bar{e},\bar{e}_{\eta},u)$.
The dynamic system of \cref{General Nonlinear System}, based on new variables, can be expressed as
\begin{equation}
    \dot{\bar{e}} = \underline{f_1}(\bar{e},\bar{e}_\eta, u), \;
    \dot{\bar{e}}_\eta =\underline{f_0} (\bar{e}, \bar{e}_\eta,u), \;
    y = C_c (u-\bar{e}_1)
    \label{eq_closedlooporigin}
\end{equation}
where in \cref{eq_closedlooporigin} $\underline{f_1}(\bar{e},\bar{e}_\eta, u) = \left[ \bar{e}_2, \,\, \bar{e}_3, \,\, \cdots, -\bar{f}_1(\bar{e}, \bar{e}_\eta, u) \right]^T$. We assume that the equilibrium $[\bar{e}^T,\bar{e}_\eta^T] =0$ is exponentially stable and, based on converse Lyapunov theorem \cite{khalil2002nonlinear}, there exists a Lyapunov function $V(\bar{e},\bar{e}_\eta)$ such that
\begin{equation}
\begin{array}{c}
c_1 \bigg\| \begin{array}{c} \bar{e}\\\bar{e}_{\eta} \end{array} \bigg\|^2 \leq V(\bar{e},\bar{e}_\eta) \leq c_2 \bigg\| \begin{array}{c} \bar{e}\\\bar{e}_{\eta} \end{array} \bigg\|^2 \\ \\

\left[ \begin{array}{cc} \frac{\partial V}{\partial \bar{e}} & \frac{\partial V}{\partial \bar{e}_\eta} \end{array} \right] 
\left[ \begin{array}{c} \underline{f_1}(\bar{e},\bar{e}_\eta, u) \vspace{0.1in} \\
\underline{f_0}(\bar{e},\bar{e}_\eta, u) \end{array} \right] \leq -c_3 \bigg\| \begin{array}{c} \bar{e}\\\bar{e}_{\eta} \end{array}  \bigg\|^2 \\ \\

\bigg\| \begin{array}{c} \frac{\partial V}{\partial \bar{e}} \vspace{0.1in} \\ \frac{\partial V}{\partial \bar{e}_\eta} \end{array} \bigg\| \leq 
c_4 \bigg\| \begin{array}{c} \bar{e}\\\bar{e}_{\eta} \end{array} \bigg\|
\end{array}
\end{equation}
where, $c_1, \, c_2, \, c_3, \, c_4 > 0$. Now consider the aforementioned system in a closed-loop structure depicted in \cref{Fig: Feedback Structure}. The closed-loop dynamic equation becomes
\begin{equation}
\begin{aligned}
    \dot{\bar{e}} &= \underline{f_1}(\bar{e},\bar{e}_\eta, u) + C^T_c\dot{u} \\
    \dot{\bar{e}}_\eta &= \underline{f_0} (\bar{e}, \bar{e}_\eta,u) + \frac{d \eta_u}{d u}\dot{u} \\
    y &= C_c\xi = \xi_1 = (u-\bar{e}_1) \\
    \dot{u} &= \lambda (r-y) = \lambda e
    \end{aligned}
\end{equation}
where $C_c$ is defined in \cref{eq_normalform}. For the closed-loop system, we define the Lyapunov function $\bar{V}(\bar{e},\bar{e}_\eta,e) = V(\bar{e},\bar{e}_\eta)+\frac{1}{2}(r-u)^2 = V(\bar{e},\bar{e}_\eta)+\frac{1}{2}(e-\bar{e}_1)^2$. By considering $r$ as a step input in the closed-loop scenario, the derivative of $\bar{V}$ along the trajectories of $\left[ \bar{e}^T \,\, \bar{e}_\eta^T \,\, e \right]$ results in,
\begin{equation*}
    \begin{aligned}
    \dot{\bar{V}} &= \dot{V}(\bar{e},\bar{e}_\eta)+(r-u)(-\dot{u})= \frac{\partial V}{\partial{\bar{e}}}\dot{\bar{e}}+
    \frac{\partial V}{\partial{\bar{e}_\eta}}  \dot{\bar{e}}_\eta - \lambda e (e-\bar{e}_1)\\
    &= \left[ \begin{array}{cc} \frac{\partial V}{\partial{\bar{e}}}& \frac{\partial V}{\partial \bar{e}_\eta} \end{array} \right] \left[ \begin{array}{c} \underline{f_1}(\bar{e},\bar{e}_\eta, u) \vspace{0.03in} \\ \underline{f_0} (\bar{e}, \bar{e}_\eta,u) \end{array} \right] + \left[ \begin{array}{cc} \frac{\partial V}{\partial{\bar{e}}}& \frac{\partial V}{\partial \bar{e}_\eta} \end{array} \right] \left[ \begin{array}{c} C^T_c \vspace{0.03in} \\ d \eta_u/{d u} \end{array} \right]\dot{u} -\lambda e^2 + \lambda e \bar{e}_1
    \end{aligned}
\end{equation*}
Hence,
\begin{equation*}
    \begin{aligned}
    \dot{\bar{V}} & \leq - c_3 \bigg\| \begin{array}{c} \bar{e}\\ \bar{e}_{\eta} \end{array} \bigg\|^2 +\lambda \left( 1 + c_4 \bigg\| \begin{array}{c} C^T_c \vspace{0.03in} \\ d \eta_u/{d u} \end{array} \bigg\| \right)\, |e| \, \bigg\| \begin{array}{c} \bar{e}\\ \bar{e}_{\eta} \end{array} \bigg\| - \lambda |e|^2 \\
    & \leq - c_3 \bigg\| \begin{array}{c} \bar{e}\\ \bar{e}_{\eta} \end{array} \bigg\|^2 +\lambda \Bigl[ 1 + c_4 \left( 1 + \big\| d \eta_u/{d u} \big\| \right) \Bigr] \, |e| \, \bigg\| \begin{array}{c} \bar{e}\\\bar{e}_{\eta} \end{array} \bigg\| - \lambda |e|^2
    \end{aligned}
\end{equation*}
Now similar to the process in \cref{Eq: proof inequlity,Eq: matrix format inequlity}, we can show that the equilibrium $\left[ \bar{e}^T \,\, \bar{e}_\eta^T \,\, e \right] = 0$ is exponentially stable for $0< \lambda < {4c_3}/{ \left[ 1+c_4\left( 1 + \big\| d \eta_u/{d u} \big\|_{max} \right) \right]^2}$. The above analysis extends the application of the compensation structure in \cref{Fig: Feedback Structure} to higher order nonlinear systems.

%%%%%%%%%%%%%%%%%%%%%%%%%%%%%%%%%%%%%%%%%%%%%%%%%%%%%
\subsection{Simulation}
We consider the following nonlinear mass-spring-damper system,
\begin{equation}
    \ddot{x}+2\zeta \omega_n \dot{x} +[1+f\left(x-u\right)]\omega_n^2 (x-u)=0, \quad y = x,
\label{eq_ex11}
\end{equation}
where $0 < \zeta < 1$, $\omega_n > 0$, and the function $f(z)$ is
\begin{equation}
    f(z) = \Bigg\{\begin{array}{cc}
         z^2 \quad \textit{for} \quad z < \beta \\
         \beta^2 \quad \textit{for} \quad z \geq \beta
    \end{array}
\label{eq_ex12}
\end{equation}
Equation (\ref{eq_ex11}) represents the motion of a damped mass which is excited via a nonlinear spring, the spring coefficient of which is given by $[1+f\left(z\right)]\omega_n^2$. To analyze the stability of the equilibrium $x = u$ (DC gain = 1), $\dot{x} = 0$ for a step input in $u$, we consider the Lyapunov function 
\begin{equation}
V = \frac{1}{2} \left( \dot{e}_1 + \zeta \omega_n e_1 \right)^2 + \frac{1}{2} e_1^2\omega_d^2 + \int_{0}^{e_1} \omega_n^2 f(z) z dz
\end{equation}
where $e_1 = x - u$, $\dot{e}_1 = \dot{x}$ for a step input in $u$ and $\omega_d = \omega_n \sqrt{1 - \zeta^2}$. Since
\begin{equation*}
    0 < \int_0^{e_1} \omega_n^2 f(z)zdz < \int_{0}^{e_1} \omega_n^2 \beta^2 zdz = \omega_n^2 \beta^2 \frac{e_1^2}{2}
\end{equation*}
we have,
\begin{equation}
\begin{array}{rcl}
    \frac{1}{2} (\dot{e}_1 +\zeta \omega_n e_1)^2 + \frac{1}{2} e_1^2 \omega_d^2 \; \leq & V & \leq \; \frac{1}{2} (\dot{e}_1+\zeta \omega_n e_1)^2 + \frac{1}{2} e_1^2 \left( \omega_d^2 + \omega_n^2 \beta^2 \right) \\
    \Rightarrow c_1 \bigg\| 
    \begin{array}{c} e_1 \\ \dot{e}_1 \end{array}
    \bigg\|^2 \leq & V & \leq c_2 \bigg\| 
    \begin{array}{c} e_1 \\ \dot{e}_1 \end{array}
    \bigg\|^2
\end{array}
\end{equation}
where it can be shown that $c_1 = 0.5\left( (1 + \omega^2_n) - \sqrt{(1 + \omega^2_n)^2 - 4\omega^2_n(1-\zeta^2)} \right) > 0$ and $c_2 = 0.5\left( 1 \cdots \right.$ $\left. + \, \omega^2_n(1+\beta^2)\right) > 0$. From \cref{eq_ex11}, using $e_1 = x - u$ and $\dot{e}_1 = \dot{x}$, we have
\begin{equation}
    \ddot{e}_1+2\zeta \omega_n \dot{e}_1 +[1+f\left(e_1\right)]\omega_n^2 e_1 =0,
\label{eq_ex13}
\end{equation}
where $f(\cdot)$ is defined in \cref{eq_ex12}. Taking the derivative of $V$ along the trajectories of \cref{eq_ex13},
\begin{equation}
\begin{array}{rcl}
    \dot{V} & = & (\ddot{e}_1 + \zeta \omega_n \dot{e}_1)(\dot{e}_1 + \zeta \omega_n e_1) + \omega_d^2 e_1\dot{e}_1 + \omega_n^2 f(e_1)e_1\dot{e}_1 \\
    & = & -\xi \omega_n [(\dot{e}_1+\zeta \omega_n e_1)^2 + e_1^2\omega_d^2] - \zeta\omega_n^3 f(e_1)e_1^2 \\
    & \le & -\xi \omega_n [(\dot{e}_1+\zeta \omega_n e_1)^2 + e_1^2\omega_d^2] < -c_3 \bigg\| 
    \begin{array}{c} e_1 \\ \dot{e}_1 \end{array}
    \bigg\|^2
\end{array}
\end{equation}
where $c_3 = \zeta\omega_n c_1 > 0$. We also note that
\begin{equation}
    \left[\begin{array}{cccc}
         \frac{\partial V}{\partial  e_1}  \vspace{0.05in} \\
         \frac{\partial V}{\partial \dot{e}_1}
    \end{array}\right] =
    M \left[\begin{array}{cc}
         e_1  \\
         \dot{e}_1 
    \end{array} \right] \leq c_4 \bigg\| \begin{array}{cc}
         e_1  \\
         \dot{e}_1 
    \end{array} \bigg\|, \quad M = \left[\begin{array}{cccc}
         \omega_n^2\left(1 + f(e_1)\right) & \zeta \omega_n\\
         \zeta \omega_n & 1
    \end{array}\right]
\end{equation}
and an estimate of $c_4$ is $c_4 = \Vert M \Vert_2 = \sigma_{max}\left( M \right) \le \sqrt{1 + \zeta^2\omega_n^2 + \omega_n^4\left(1 + \beta^2 \right)^2}$, where $\sigma_{max}\left( M \right)$ is the maximum singular value of $M$. From the above analysis, we conclude that for a step input in $u$, $e_1 = \dot{e}_1 = 0$ is an exponentially stable equilibrium, see \cite{khalil2002nonlinear}. Equation (\ref{eq_ex11}) belongs to the category of systems considered in \cref{sec_high} with $\xi_1 = y = x$ and $\xi_2 = \dot{x}$ and a relative degree $\rho = 2$, implying there are no internal state $\eta$. We concluded in \cref{sec_high} that when it is placed in the feedback-feedforward configuration of \cref{Fig: Feedback Structure}, it has an exponentially stable equilibrium at $y = r = u$ and $\dot{e}_1 = 0$ for the integral gain $\lambda$ satisfying $0 < \lambda < {4c_3}/{ \left[ 1+c_4\left( 1 + \big\| d \eta_u/{d u} \big\|_{max} \right) \right]^2} = {4c_3}/{\left[ 1+c_4 \right]^2}$, since there is no internal state $\eta$.

To demonstrate the effect of $\lambda$ on the closed-loop stability, we simulate the system in \cref{eq_ex11} with $\omega_n = 1$rad/s, $\zeta= 0.7$ and $\beta^2 = 20$. 
\begin{figure}[htbp]  
	\begin{center} 
		\includegraphics[width=0.95\textwidth]{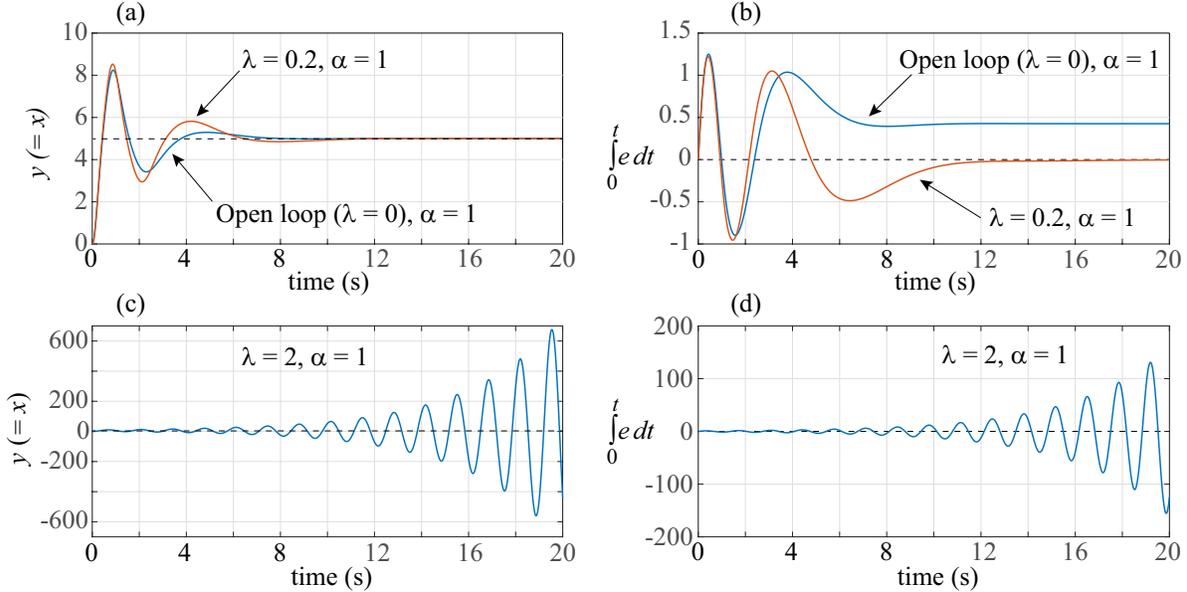}
	\end{center}
	\vspace{-0.25in}
    \caption{Step response and $e_r = \lim_{t\rightarrow \infty} \int_{0}^{t} e(\tau)d\tau$ for \cref{eq_ex11} under compensation of \cref{Fig: Feedback Structure}}
    \vspace{-0.05in}
	\label{Fig: ExpExampleTwoStateSystem}
	\vspace{-0.0in}
\end{figure}
The analysis above gives $c_1 = 0.3$, $c_3 = 0.21$, and $c_4 \approx 21$, yielding $0 < \lambda < 0.0017$ as a sufficient condition for closed-loop stability. This is a conservative estimate of $\lambda$. However, from linearization of the closed-loop system around $x = r$, which is,
\begin{equation*}
    \dddot{x} + 2\zeta\omega_n\ddot{x} + \omega_n^2\dot{x} + \omega_n^2\lambda(x - r) = 0
\end{equation*}
the necessary and sufficient condition for local stability, using the Routh-Hurwitz criterion, is $0 < \lambda < 2\zeta\omega_n = 1.4$. Simulation results are depicted in \cref{Fig: ExpExampleTwoStateSystem} for the open-loop case (i.e. $\lambda=0$) and for $\lambda = 0.2, 2$, with $\alpha = 1( = 1/\mbox{DC gain})$.

In \cref{Fig: ExpExampleTwoStateSystem}(a), the step response of the closed-loop system, when $\lambda=0.2$, exponentially converges to step input $r = 5$, as does the open-loop system ($\lambda = 0$). In addition, \cref{Fig: ExpExampleTwoStateSystem}(b) illustrates that $\int_{0}^{\infty} e(\tau)d\tau \ne 0$ for $\lambda = 0$, and $\int_{0}^{\infty} e(\tau)d\tau=0$ for $\lambda = 0.2$. However, when $\lambda=2$ the output and integral constraint values diverge, as shown in \cref{Fig: ExpExampleTwoStateSystem}(c),(d), since the system is unstable for $\lambda=2 > 1.4$.

%%%%%%%%%%%%%%%%%%%%%%%%%%%%%%%%%%%%%%%%%%%%%%%%%%%%%%%%%%%%
\vspace{-0.15in}
\section{Conclusion}
\vspace{-0.1in}
This paper demonstrates a method of shaping the transient step response of nonlinear systems using the compensation structure of \cref{Fig: Feedback Structure}, to satisfy a class of integral constraints. The method is applicable to SISO nonlinear systems that provide a stable step response and have a positive DC gain. Such nonlinear systems are plentiful, including hybrid energy systems where the aforementioned integral constraints are key to power management. First, the ability of the compensated system to satisfy the integral constraints while tracking step inputs is proven. Thereafter, analysis of first order nonlinear plants utilizing {\it Converse Lyapunov Theorems} establishes bounds on the integrator gain $\lambda$ where the closed-loop system is stable. The analysis assumes that the nonlinear plant has an exponentially or asymptotically stable response to step inputs, but does not require a detailed plant model. For higher order SISO nonlinear systems, the {\it Normal Form} is used to convert the original system to a chained-form of integrators and an internal dynamic. Here again, by assuming broad stability characteristics of the unknown plant, {\it Converse Lyapunov Theorems} are used to establish conditions for stability of the compensated system. Simulations demonstrate the efficacy of the study.

\end{document}